\documentclass{llncs}
\usepackage{makeidx}
\usepackage{paralist}
\usepackage{amsmath}
\usepackage{algorithm}
\usepackage{booktabs} 
\usepackage{amssymb}
\usepackage{amsmath}
\usepackage{algorithm}
\usepackage{graphicx}
\usepackage[noend]{algpseudocode} 
\usepackage{xcolor,colortbl}
\usepackage{tabularx}
\usepackage{array}
\newcolumntype{L}[1]{>{\raggedright\let\newline\\\arraybackslash\hspace{0pt}}m{#1}}
\newcolumntype{C}[1]{>{\centering\let\newline\\\arraybackslash\hspace{0pt}}m{#1}}
\newcolumntype{R}[1]{>{\raggedleft\let\newline\\\arraybackslash\hspace{0pt}}m{#1}}
\begin{document}
\frontmatter          
\pagestyle{headings}  
%
\title{Fast Nearest-Neighbor Classification using RNN in Domains with Large Number of Classes}
%
%

\author{
Gautam Singh\inst{1} \and
Gargi Dasgupta\inst{2} \and
Yu Deng\inst{3}
}
\authorrunning{Gautam Singh et al.} 
%
\tocauthor{Gautam Singh, Gargi Dasgupta, Yu Deng}
\institute{
IBM Research-India, New Delhi DL 110070, India,\\
\email{gautamsi@in.ibm.com}
\and
IBM Research-India, Bangalore KA 560045, India,\\
\email{gaargidasgupta@in.ibm.com}
\and
IBM T.J. Watson Research Center, Yorktown Heights, New York NY 10598, US,\\
\email{dengy@us.ibm.com}
}

\maketitle              

\begin{abstract}
In scenarios involving text classification where the number of classes is large (in multiples of 10000s) and training samples for each class are few and often verbose, nearest neighbor methods are effective but very slow in computing a similarity score with training samples of every class. On the other hand, machine learning models are fast at runtime but training them adequately is not feasible using few available training samples per class.
In this paper, we propose a hybrid approach that cascades 1) a fast but less-accurate recurrent neural network (RNN) model and 2) a slow but more-accurate nearest-neighbor model using bag of syntactic features. 

Using the cascaded approach, our experiments, performed on data set from IT support services where customer complaint text needs to be classified to return top-$N$ possible error codes, show that the query-time of the slow system is reduced to $1/6^{th}$ while its accuracy is being improved. Our approach outperforms an LSH-based baseline for query-time reduction.
We also derive a lower bound on the accuracy of the cascaded model in terms of the accuracies of the individual models. In any two-stage approach, choosing the right number of candidates to pass on to the second stage is crucial. We prove a result that aids in choosing this cutoff number for the cascaded system. 



\end{abstract}

\section{Introduction}

In the spectrum of text classification tasks, number of class labels could be two (binary), more than two (multi-class) or a very large number (more than 10000). In industry, text classification tasks with large number of classes arise naturally. In the domain of IT customer support, a user complaint text is classified to return top-$N$ most likely error codes (from potentially 10000s of options) that product could be having. Another example is from the domain of health insurance where patients inquire  whether their insurance covers a certain diagnosis or treatment. Such patient queries need to be classified into top-$N$ appropriate medical codes to look up against a database and serve an automated response to the patient.

In the given setting, enough training samples are not available to adequately train an effective ML-based model \cite{tsoumakas2008effective}. For dealing with this challenge, one work \cite{dumais2000hierarchical} proposes a hierarchical classification where a hierarchy among class labels is known before-hand. Example, an item is classified into top-level categories (``computer" or ``sports") and then further classified into sub-categories (``computer/hardware", ``computer/software" etc.). In another approach \cite{tsoumakas2008effective}, the hierarchy among classes is not known. Instead, from the ``flat" class labels, a hierarchy is constructed through repeated clustering of the classes. 

In this paper, we adopt a different approach. As the number of class labels grows, the task of text classification starts to increasingly resemble the task of document retrieval (or search). Our approach makes use of this observation. Retrieval methods using sophisticated features are effective but very slow at prediction time. ML models on the other hand are fast but imprecise in the given setting. A common approach in retrieval domain uses two-stages 1) \emph{filtering stage}, a fast, imprecise and inexpensive stage that generates $t$ candidate documents and 2) \emph{ranking stage}, a sophisticated retrieval module that uses complex features (phrase-level or syntactic-level) to re-rank the candidate documents. The two stage retrieval approach mitigates the trade-off between speed and accuracy. By analogy, in this paper, we use statistical, ML based model as the first stage (i.e. candidate generation). This stage is fast but has low accuracy. Next, we use expensive syntactic NLP features and similarity scoring on the candidate classes in the second-stage to generate final top-$N$ predicted classes. This stage is slow but more accurate. 

A number of ML-based models exist for text classification such as regression models \cite{schutze1995comparison}, Bayesian models \cite{mccallum1998comparison} and emerging deep neural networks \cite{rnn,lstm,cho,zhang2015character,zhou2016text}. On the other hand, many approaches use syntactic NLP-based features for text classification based on similarity of nearest neighbor \cite{sidorov,wang}. Another approach uses \emph{word2vec} to incorporate word similarity into nearest-neighbor-based text classification task \cite{xin}. For candidate generation, hashing has been a well-known technique. Hashing techniques can either be \emph{data-agnostic} (such as locality sensitive hashing \cite{andoni2006near,broder2000min}) or \emph{data-dependent} such as learning to hash \cite{wang2017survey}. Candidate generation is also classified as \emph{conjunctive} if the candidates returned contain all the terms in query and \emph{disjunctive} if the candidates contain at least one term from the query \cite{clarke2016assessing,asadi2013effectiveness}.

%

The main contributions of this paper are as follows: 
\begin{asparaenum}
\item We propose a hybrid model for text classification that cascades a fast but less-accurate recurrent neural network model and a slow but more-accurate retrieval model which uses bag of syntactic features. We experimentally show that the query time of the slow-retrieval model is reduced to $1/6^{th}$ after cascading while improving upon its accuracy. (Section \ref{sec:cascaded_approach})
\item We prove a meaningful lower bound on the accuracy of cascaded model in terms of the accuracies of the individual models. The result is generic and can be applied on any cascaded retrieval model. (Section \ref{sec:cascaded_approach_lowerbound})
\item Choosing the number of candidate classes $t$ to pass on to the second stage in any cascaded model is crucial to the performance. If $t$ is too small, the accuracy of the second stage suffers. If $t$ is too large, the speed of the model suffers. To choose this, past works typically perform grid search or test on $t$ values at regular intervals within a desirable range. In this paper, we prove that in order to choose the best $t$, we need to test the accuracy of the cascaded model only on few special values of $t$ rather than all possible values within a desirable range. The result is generic and can be applied to any cascaded retrieval model. (Section \ref{sec:cascaded_approach_pickt})
\item We show that our cascaded model outperforms a baseline for speeding-up the slow retrieval model using locality sensitive hashing (LSH). (Section \ref{sec:experiments})
\end{asparaenum}



\section{Nearest-Neighbor Model using Bag of Syntactic Features}
In this section, we describe the aforementioned slow nearest-neighbor based  model. In this technique, we first perform dependency parsing on the text. A dependency parser, $depparse(s)$ takes a sentence $s$ as input and returns a tree $(V,E)$
\begin{align}
V,E = depparse(s)
\end{align}
where $V$ is a set of all nodes or words in sentence $s$ and $E$ is a set of
all edges or 3-tuples in the tree.
\begin{align}
E = \{(w_1,w_2,r) | \exists \text{ directed edge $r$ from $w_1\rightarrow
w_2$}\}
\end{align}
Any directed edge $(w_1,w_2,r)\in E$ represents some grammatical relation $r$
between the connected words $w_1$ and $w_2$. These relations might have labels
such as \emph{nsubj, dobj, advmod,\ldots etc.} and these represent some grammatical function fulfilled by the connected word pair.

Next, we take word-pairs using each edge in the dependency tree and concatenate their word vectors \cite{mikolov} to get a bag of \emph{syntactic} feature vectors. This is shown in Algorithm \ref{algo:synbigram}. In the algorithm, notice that weights are assigned to the words during concatenation. These weights are based on heuristics and higher weight is given to nouns, adjectives and verbs than other parts of speech.

\begin{algorithm}
\caption{Generate Bag of Syntactic Features}
\begin{algorithmic}[1]
\Procedure{GenerateBagOfSyntacticFeatures}{$s$}
\State $(V,E) \gets \text{dependency parse tree of sentence } s$ 
\State $vectorset \gets \text{empty set}$
\For {$(w_1,w_2,rel) \in E$}
\State $\mathbf{v}_1 \gets$ $word2vec(w_1)$
\State $\mathbf{v}_2 \gets$ $word2vec(w_2)$
\State $\omega_1 \gets$ $wordweight(w_1)$
\State $\omega_2 \gets$ $wordweight(w_2)$
\State $\mathbf{v} \gets (\omega_1\mathbf{v}_1, \omega_2\mathbf{v}_2)$
\State $vectorset.add(\mathbf{v})$
\EndFor
\Return $vectorset$
\EndProcedure
\end{algorithmic}
\label{algo:synbigram}
\end{algorithm}

\subsubsection{Computing similarity of text with a class in training set}
Given a text query $q$, we next compute its similarity with a particular class $b$ in training set. Let the set of texts in the training set corresponding to class $b$ be called $X(b)$. Let the $Z(b)$ denote the set of bags of syntactic features corresponding to each text in $X(b)$. Let $z(q)$ denote bag of syntactic features for the query text.

\begin{align}\label{eq:snbigramscore}
sim(q,b)=\max\limits_{vecset \in Z(b)}\sum\limits_{\mathbf{u}\in z(q)}\max\limits_{\mathbf{v}\in vecset}
cosine(\mathbf{u},
\mathbf{v})
\end{align}
In the above similarity metric, we compute the cosine of feature vectors of the query text and texts corresponding to class $b$ in the training set. The similarity of the best matching text is taken as the similarity score for class $b$. Next, the $N$ highest scoring classes for the given query are returned.

\section{Recurrent Models for Text Classification}
For text classification using recurrent models, text is converted into a sequence of word-vectors and given as input to the model. In recurrent models, the words in text may be processed from left to right. In each iteration, previous hidden state and a word are processed to return a new hidden state. In this paper, we experiment with two kinds of recurrent models 1) GRU \cite{cho} and 2) LSTM \cite{lstm}. We describe below the details only for the GRU model. 

GRU model is parametric and defined by 6 matrices $\mathbf{U}^z$, $\mathbf{U}^r$, $\mathbf{U}^h$, $\mathbf{W}^z$,$\mathbf{W}^r$,$\mathbf{W}^h$ and output matrix $\mathbf{O}$. The recurrence equations are given below.
\subsubsection*{Initialization}
Initialize the hidden state as a zero vector.
\[
\mathbf{h}_0 = \mathbf{0}
\]
\subsubsection*{Iteration}
For $j^{th}$ iteration, $j\in[1,m]$, compute the following
\[
\mathbf{z}_j = \sigma ( \mathbf{U}^z\mathbf{v}_j + \mathbf{W}^z\mathbf{h}_{j-1}) 
\] 
\[
\mathbf{r}_j = \sigma ( \mathbf{U}^r\mathbf{v}_j + \mathbf{W}^r\mathbf{h}_{j-1}) 
\] 
\[
\mathbf{h'}_j = tanh(\mathbf{U}^h\mathbf{v}_j +
\mathbf{W}^h(\mathbf{h}_{j-1}\circ\mathbf{r}_j))
\]
\[
\mathbf{h}_j = (1-\mathbf{z}_j)\circ\mathbf{h'}_j +
\mathbf{z}_j\circ\mathbf{h}_{j-1}
\]

where $m$ is the number of words in text and $\sigma$ refers to the sigmoid function.

\subsubsection*{Termination and Computing Output Probability Distribution}
The latest hidden state $\mathbf{h}_m$ is subjected to a softmax layer to generate an output probability distribution $\mathbf{y} = softmax(\mathbf{O}\mathbf{h}_m)$. We return the classes corresponding to top-$N$ probability values in $\mathbf{y}$.

\section{Cascaded Model for Fast and Accurate Retrieval}
\label{sec:cascaded_approach}
Retrieval model using bag of syntactic features is an example of nearest-neighbor classification. For a given query $q$, this demands that the similarity score be computed with every sample in the training set. On the contrary, if we filter a few candidate classes using the first stage of cascading, the slowness of the retrieval model is overcome. We denote the recurrent machine learning model as $M_1$ and the slow nearest-neighbor classifier as $M_2$.

\subsubsection*{Notations} The correct class to which query $q$ belongs is denoted by $b_q$.  We denote the set of candidate classes returned by the first stage by $T$ and number of such candidates by $t=|T|$. We use $M_1(q,t)$ denote the set of $t$ classes returned by the first stage. We denote the number of classes to be returned by the second stage as $N$. Therefore $t>N$. $M_2^T(q,N)$ denotes the set of $N$ classes returned by the second stage after inspecting the set of classes $T$ returned by the first stage. We use $M_2(q,N)$ to denote the set of $N$ classes returned by the second stage if it were to inspect all classes in the training set without any cascading. We define an empirical accuracy metric over a validation set $S_{valid}$ containing user-queries as follows.
\begin{align}
accuracy = \frac{|\{q \mid q\in S_{test}, b_q\in 
M(q, N)\}|}{|S_{test}|} \approx P(b_q \in M(q,N))
\end{align}
The numerator is the number of queries with correct classes in the top-$N$ suggestions returned by text classifier $M$. The denominator is the total number of queries.

Before describing the proofs, we define two empirical quantities $\rho(t)$ and $\alpha(t)$ associated to the cascaded model which are easy to estimate as follows using a validation set.
\begin{align}
\begin{split}
\label{rho_estm}
\rho(t) &= \frac{|\{ q \mid q \in S_{valid}, b_q \in M_1(q,t) , b_q \in
M_2(q,N)\}|}{|\{ q \mid q \in S_{valid}, b_q \in
M_2(q,N)\}|} \\
&\approx P(b_q \in M_1(q,t) \mid b_q \in M_2(q,N))
\end{split}
\end{align}
\begin{align}
\label{alpha_estm}
\alpha(t) &= \frac{|\{ q \mid q \in S_{valid}, b_q \in M_1(q,t)\}|}{|\{ q \mid q \in S_{valid}\}|}\approx P(b_q \in M_1(q,t))
\end{align}
It is easy to compute $\rho(t)$ as follows. $\alpha(t)$ is analogously computed.
\begin{asparaenum}
\item Run both $M_1$ and $M_2$ on the validation set and store the match scores for each class.
\item For each $t$, find the number of classes which are present both in top-$t$ for $M_1$ and top-$N$ for $M_2$. Also find the number of classes which are present in top-$N$ for $M_2$.
\item Find the ratio of the above two numbers for each $t$.
\end{asparaenum}
In this paper, we assume that empirical estimates of probability values using the validation set are good approximations of their actual values.
\subsection{Lower bound on accuracy of cascaded model}
\label{sec:cascaded_approach_lowerbound}
The idea is to show that the accuracy of the cascaded model is lower bounded by accuracy of the slow-model times $\rho(t)$. This is given in following theorem.
\begin{theorem} 
\label{theorem:th1}
For a cascaded model consisting of stages $M_1$ and $M_2$,
\begin{align}
P(b_q \in M_1(q,t) , b_q \in M_2^{T}(q,N)) \geq \rho(t) P(b_q \in M_2(q,N))
\end{align}
\end{theorem}
In order to prove the above, we go through the following lemma.
\begin{lemma}
\label{lemma:small} 
Let $q$ by any query such that $b_q \in M_1(q,t)$, then
\begin{align}
b_q \in M_2(q,N) \implies b_q \in M_2^T(q,N)
\end{align}
\end{lemma}
\begin{proof}[of Lemma \ref{lemma:small}] Since $b_q \in M_1(q,t)$, hence the first stage removes only some incorrect classes and not the correct class $b_q$. Now since the correct class is in top-$N$ for $M_2$ without any cascading, hence, after cascading using $M_1$, the candidate classes that $M_2$ inspects contain fewer incorrect classes. Thus, introduction of cascading either improves or maintains the rank of the correct class returned in top-$N$. This gives us the above lemma.
\qed
\end{proof}
\begin{proof}[of Theorem \ref{theorem:th1}]
Using Lemma \ref{lemma:small}, 
\begin{align}
b_q \in M_1(q,t) \cap M_2(q,N) \implies b_q \in M_1(q,t) \cap M_2^T(q,N)
\end{align}
This implies that,
\begin{align}
\begin{split}
P(b_q \in M_1(q,t) \cap M_2^T(q,N)) &\geq P(b_q \in M_1(q,t) \cap M_2(q,N))\\
&= P(b_q \in M_1(q,t), b_q \in M_2(q,N))\\
&= P(b_q \in M_1(q,t) \mid b_q \in M_2(q,N)) P(b_q \in M_2(q,N))\\
&= \rho(t)P(b_q \in M_2(q,N))
\end{split}
\end{align}
\qed
\end{proof}
\subsection{Picking best $t$}
\label{sec:cascaded_approach_pickt}
Given a cascaded model, choosing the number of candidates $t$ to pass on to the second stage is crucial. If $t$ is too small, then accuracy suffers as it becomes more likely that the correct class has not passed the first stage. If $t$ is too large, then the query time suffers. The first stage also acts as an elimination round and large $t$ dilutes this elimination process by crowding out the correct class.

Text classification models used in each stage are typically complex. Studying their combined behavior in a cascaded setting may not be straightforward. Thus, choosing $t$ is a challenge. Typically, the only reliable way to do this is to run the cascaded model on all possible values of $t$ and pick a $t$ which produces the highest accuracy on a validation set within a desirable range of $t$. This process might be time-consuming as the slow model (as a part of cascaded model) needs to be re-run for every $t$ being checked. In the following theorem, we show that not all values of $t$ need to be checked. Given that $\alpha(t)$ has same value for two distinct values of $t$, the theorem shows that choosing the smaller value of $t$ offers at least as much accuracy as choosing the larger one. This implies that we need to check only those values of $t$ where $\alpha(t)$ changes value.

\begin{theorem}
\label{theorem:pickt}
Let $q$ be any query. For $t_1,t_2$ such that $t_2 > t_1$ , if $\alpha(t_1) = \alpha(t_2)$ then
\begin{align}
\begin{split}
\label{theoremequation}
& P(b_q \in M_1(q,t_1) \cap M_2^{T_1}(q,N)) \geq P(b_q \in M_1(q,t_2) \cap M_2^{T_2}(q,N))
\end{split}
\end{align}
In other words, for a given value of $\alpha(t) = \alpha_0$, the accuracy is maximized when
\begin{align}
t & = \text{arg} \min\limits_{\alpha(t) = \alpha_0} \alpha(t)
\end{align}
\end{theorem}
For proof of above theorem, we go through the following lemma.
\begin{lemma}
\label{lemma:m1ordering}
$\forall t_1,t_2$ such that $t_2 > t_1$, 
\begin{align}
b_q \in M_1(q,t_1) \implies b_q \in M_1(q,t_2)
\end{align}
\end{lemma}
\begin{proof}[of Lemma \ref{lemma:m1ordering}]
If the correct class is returned in top-$t_1$ by the first stage for a given query, then for $t_2 > t_1$, the correct class is also a part of top-$t_2$ classes returned by the first stage.
\qed
\end{proof}
\begin{proof}[of Theorem \ref{theorem:pickt}]
We start from the condition given in the theorem i.e., $\alpha(t_1) = \alpha(t_2)$ and using Equation \ref{alpha_estm}, we get
\begin{align}
\label{eq:equalprob}
P(b_q \in M_1(q,t_1)) = P(b_q \in M_1(q,t_2))
\end{align}
Using above Equation \ref{eq:equalprob} and Lemma \ref{lemma:m1ordering}, we get the equality of the set of queries for whom the correct class have passed through the first stage.
\begin{align}
\label{eq:equalcandidates}
\{q \mid b_q \in M_1(q,t_1) \} \equiv \{q \mid b_q \in M_1(q,t_2) \}
\end{align}
Now consider the set of queries which are correctly classified in top-$N$ by the cascaded model using $t = t_2$,
\begin{align}
\{q \mid b_q \in M_1(q,t_2) \cap M_2^{T_2}(q,N)\}
\end{align}
Now, when $t_2$ is reduced to $t_1$, we know on one hand that the number of classes passing to the second stage is smaller i.e., $t_1$. On the other hand, we know from set equivalence in Equation \ref{eq:equalcandidates} that the exact same queries contain their correct classes in the candidate classes being passed on. These two observations imply that only incorrect classes have been removed in the first stage while going from $t_2$ to $t_1$. This reduction in number of classes being passed on can either improve or keep same the rank of the correct class returned in the top-$N$ by the second stage in the cascaded setting. Therefore,
\begin{align}
b_q \in M_1(q,t_2) \cap M_2^{T_2}(q,N) \implies b_q \in M_1(q,t_1) \cap M_2^{T_1}(q,N)
\end{align}
This implies that
\begin{align}
P(b_q \in M_1(q,t_1) \cap M_2^{T_1}(q,N)) \geq P(b_q \in M_1(q,t_2) \cap M_2^{T_2}(q,N))
\end{align}
\qed
\end{proof}

\subsection{Baseline for Query Time Improvement}
This section describes the LSH-based baseline for candidate generation. In the training set, for every syntactic feature vector $\mathbf{v}$, the $i^{th}$ bit of the hash code is given as
\begin{align}
\label{eq:lshprojection}
h_i(\mathbf{v}) = sgn(\mathbf{w}_i^T\mathbf{v})
\end{align}
where $\mathbf{w}_i$ are randomly picked. We create a hash-table $\Phi_{class}$ whose indices are hash-codes of syntactic features in the training set and values are the sets of corresponding class labels. We create a similar hash-table $\Phi_{text}$ whose values are texts corresponding to the hash-codes instead of class labels. For candidate generation, we use two implementations of the conjunctive approach 1) \emph{Class-based} where returned candidate classes contain to all hash-codes computed from the query text. 2) \emph{Text-based} where returned candidate classes have at least one text that contains all hash codes computed from the query text.

\section{Experiments and Inferences}
\label{sec:experiments}
In this section, we describe the experiments which demonstrate
performances of our proposed techniques and verify the bounds.
\subsubsection*{Data Set}
Two kinds of documents from the domain of IT support are used to generate data set for our experiments 1) product
reference documents and 2) past problem requests. From 300MB
of product reference documents, we extracted a total of 55K
distinct error codes and a total of 15K distinct error code text descriptions. We combined the error codes corresponding to each of 15K distinct error code descriptions to reduce data sparsity per class and to get 15K error-code classes. From the past problem requests, we extracted 40K problems with known error-code classes. Out of these, 90\% are used for training while remaining is set aside
for validation and testing. Notice that the mean number of texts corresponding to each error code class is approximately 2-3, which is too few for adequate training of statistical ML-based models.\\ 
\begin{table*}[t]
\scriptsize
\centering
\begin{tabular}{L{0.2\textwidth} L{0.8\textwidth}}
\toprule
User Query & Top-$10$ Error Description Suggestions\\
\midrule
getting media err detected on device & system lic detected a program exception, 
a problem occurred during the ipl of a partition ,
partition firmware detected a data storage error ,
tape unit command timeout,
interface error,
\textbf{tape unit detected a read or write error on tape medium},
tape unit is not responding,
an open port was detected on port 0 ,
contact was lost with the device indicated ,
 destroy ipl task\\
\midrule
hmc appears to be down & \textbf{licensed internal code failure on the hardware
management console hmc} , system lic detected a program exception ,
service processor was reset due to kernel panic ,
the communication link between the service processor and the hardware management
console  hmc  failed , platform lic detected an error,
power supply failure,
processor 1 pgood fault  pluggable ,
system power interface firmware  spif  terminated the system because it detected a power fault ,
detected ac loss,
a problem occurred during the ipl of a partition ,
platform lic failure\\
\midrule
failed power supply & a fatal error occurred on power supply 1,
power supply failure,
the power supply fan on un e1 failed , \textbf{detected ac loss},
a fatal error occurred on power supply 2,
power supply non power fault  ps1 ,
the power supply fan on un e2 failed  the power supply should be replaced as
soon as possible , a non fatal error occurred on power supply
1, the power supply fan on un e2 experienced a short stoppage\\
\bottomrule
\end{tabular}
\vspace{8px}
\caption{Examples of user-queries and error-code class descriptions returned by our models with highlighted correct response}
\label{examples_table}
\end{table*} 

\begin{figure*}[t!] 
\centering
\includegraphics[width=1.0\linewidth]{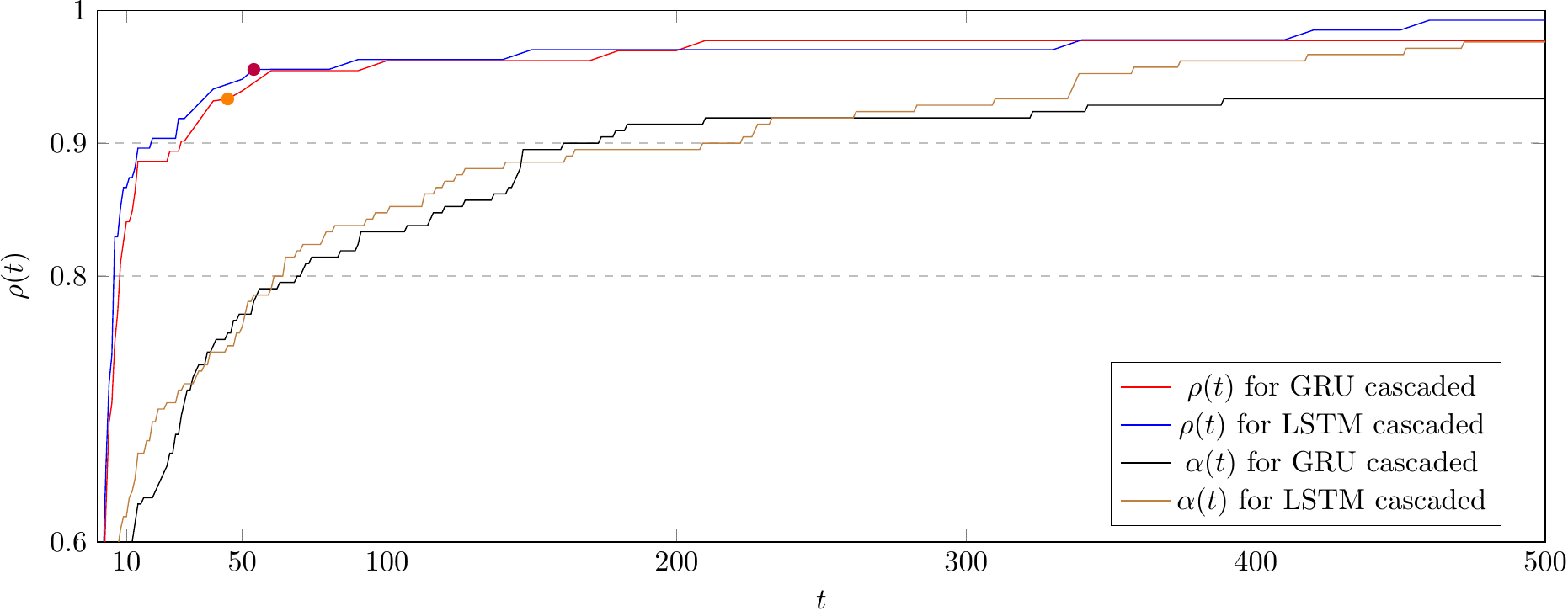}
\caption{Plot showing dependence of $\rho(t)$ and $\alpha(t)$ on $t$ for the cascaded model on the described data set. }
\label{rho_plot}
\end{figure*}

\begin{table}[t!]
\begin{minipage}{.45\linewidth}
\centering
\scriptsize
\begin{tabular}{c c c c c}
\toprule
$M_1$ & $\frac{P(b_q \in M_1(q,N))}{P(b_q \in
M_2(q,N))}$ & $t$  & $t$ values when $\rho$\\
& & crossing & changes value\\
\midrule
GRU & 0.933 & 45 & 45,55,91,179,210\\
LSTM & 0.955 & 54 & 54,82,141,337,452\\
\bottomrule
\end{tabular}
\vspace{8px}
\caption{Computing relevant $t$ values for cascaded models}
\label{tab:alpha_rho_analysis} 
\end{minipage}
\hspace{0.09\linewidth}
\begin{minipage}{.45\linewidth}
\centering
\scriptsize
\begin{tabular}{ccccccc}
\toprule
My\cellcolor{blue!47} & tape \cellcolor{blue!58} & drive \cellcolor{blue!52} & has \cellcolor{blue!27} & been \cellcolor{blue!33} &giving\cellcolor{blue!30} \\ error\cellcolor{blue!27} & once\cellcolor{blue!31} &
every\cellcolor{blue!29} & week\cellcolor{blue!51}\\
\midrule
Need\cellcolor{blue!37}  & to\cellcolor{blue!26}  & have \cellcolor{blue!27}  & adapter \cellcolor{blue!52} & replaced\cellcolor{blue!46} \\
\midrule
Flashing \cellcolor{blue!61}& power\cellcolor{blue!57} & button\cellcolor{blue!55} & and\cellcolor{blue!31} & warning\cellcolor{blue!41} & light\cellcolor{blue!55}\\
\bottomrule
\end{tabular}
\vspace{8px}
\caption{Heat map showing word weights assigned by GRU model to user queries}
\label{tab:gru_wordweights} 
\end{minipage}
\end{table}

\begin{table}[t!]
\begin{minipage}{.45\linewidth}
\centering
\scriptsize
\begin{tabular}{c c c c c c}
\toprule
$M_1$ & $t$ & Accuracy & \multicolumn{3}{c}{Time
Taken (in s)}\\
\cmidrule(l){4-6}
& & & Mean & Min & Max\\
\midrule
LSTM & 54 & 63.33\% & 9.09 & 0.73 & 49.98\\
& 82 & 64.76\% & 12.23 & 1.12 & 63.60\\
& 141 & \textbf{65.23\%} & 14.51 & 1.23 & 75.00\\
& 337 & 64.29\% & 23.35 & 2.09 & 114.4\\
& 452 & 62.86\% & 23.46 & 2.35 & 108.2\\
\midrule
GRU & 45 & 60.95\% & 9.85 & 1.22 & 50.23\\
& 179 & 63.33\% & 18.93 & 2.25 & 89.97\\
& 210 & \textbf{64.29\%} & 20.10 & 2.31 & 93.18\\
& 400 & 63.81\% & 25.89 & 2.51 & 117.1\\
& 500 & 63.33\% & 29.15 & 2.66 & 131.6\\
\bottomrule
\end{tabular}
\vspace{8px}
\caption{Accuracies and CPU times of cascaded model
comprising of syntactic-bigram vector model followed by $M_1$ for
varying $t$ for suggestion of top-10 error-code classes}
\label{tab:cascaded_accuracies} 
\end{minipage}
\hspace{0.09\linewidth}
\begin{minipage}{.45\linewidth}
\centering
\scriptsize
\begin{tabular}{c c c c c}
\toprule
Model & Accuracy & \multicolumn{3}{c}{Time
Taken (in s)}\\
\cmidrule(l){3-5}
& & Mean & Min & Max\\
\midrule
LSTM & 61.43\% & 0.013 &
0.004 & 0.112\\
GRU & 60.00\% & 0.014 &
0.005 & 0.062\\
\midrule
sn-Vectors & 64.29\% & 84.60 & 12.54 & 294.92\\
sn-Bigrams & 63.33\% & 41.42 & 8.97 & 133.51\\
BOW & 43.33\% & 12.80& 7.32 & 28.19\\
\bottomrule
\end{tabular}
\vspace{8px}
\caption{Accuracies and CPU times of various models for suggestion of top-10 error-code classes}
\label{tab:uncascaded_accuracies}
\end{minipage}
\end{table}

\begin{table}[t!]
\centering
\begin{tabular}{c c c c c c}
\toprule
LSH & $P$ & Accuracy & \multicolumn{3}{c}{Time
Taken (in s)}\\
\cmidrule(l){4-6}
Version& & & Mean & Min & Max\\
\midrule
Cluster & 5 & 62.38\% & 35.1 & 0.001 & 142.2\\
based& 10 & 60.47\% &	10.3	& 0.001 &	27.12\\
& 15 &	57.14\% &	5.82	& 0.001	& 19.39\\
& 20	& 56.66\% &	5.79 &	0.001 &	19.75\\
\midrule
Text & 1 &	64.28\% &	46.1 &	0.001 &	194.2\\
based& 3 &	61.90\%	& 13.5 &	0.001	& 41.23\\
& 5	& 55.23\% &	1.31 &	0.001 &	11.30\\
\bottomrule
\end{tabular}
\vspace{8px}
\caption{Baseline accuracies (in top-10) and CPU times of LSH-based implementations (for reducing query time of bag of syntactic features technique) for varying number of permutations $P$.}
\label{tab:lsh_accuracy_table} 
\end{table}

In Figure \ref{rho_plot}, we show the plots of $\rho(t)$ and $\alpha(t)$ for cascaded models having $M_1$ as the GRU model and the LSTM model. Notice that to guarantee the usefulness of the cascaded model, the accuracy of the cascaded model should be at least as much as the less accurate model. The smallest value of $t$ that achieves this can be found by using the lower bound in Theorem \ref{theorem:th1}. Thus,
\begin{align}
\rho(t)P(b_q \in M_2(q,N)) \geq P(b_q \in M_1(q,N))
\end{align}
or,
\begin{align}
\rho(t) \geq \frac{P(b_q \in M_1(q,N))}{P(b_q \in M_2(q,N))}
\end{align}
Thus, as shown in Table \ref{tab:alpha_rho_analysis}, for GRU model, $\rho(t) \geq 0.933$ or $t \geq 45$. Similarly, for LSTM model, $\rho(t) \geq 0.955$ or $t \geq 54$. These thresholds on $t$ are shown in Figure \ref{rho_plot} on the $\rho(t)$ plots.

In Table \ref{tab:cascaded_accuracies}, we show the accuracies and CPU times of the cascaded model for varying $t$ for $M_1$ as the GRU and the LSTM model. Notice that according to Theorem \ref{theorem:pickt}, we only need to check the accuracies for $t$ where $\alpha(t)$ (shown in Figure \ref{rho_plot}) changes value. Therefore Table \ref{tab:cascaded_accuracies} shows accuracies for some of those $t$ values.

On comparing the accuracy of CPU times of the cascaded model (in Table \ref{tab:cascaded_accuracies}) and the bag of syntactic features model (depicted in Table \ref{tab:uncascaded_accuracies} as sn-Vectors), we see that cascading reduces query time to $1/6^{th}$ using LSTM and $1/4^{th}$ using GRU model. Cascading using LSTM model also improves the accuracy.

Table \ref{tab:uncascaded_accuracies} shows accuracy and CPU times of other uncascaded models such as 1) fast, machine learning based LSTM and GRU models, 2) bag of syntactic bigrams which uses exact string match for finding similarities between syntactic-bigrams after lemmatizing the words and 3) the bag of words model.

In Table \ref{tab:lsh_accuracy_table}, we show the results of the two versions of the LSH based baseline. The accuracy and CPU time are shown for varying number of permutations (number of bits in the hash code). Increasing the number of permutations leads to fewer nearest-neighbor candidates which decreases the accuracy and improves query time. Comparing results in Table \ref{tab:cascaded_accuracies} and \ref{tab:lsh_accuracy_table}, we infer that our proposed cascaded model outperforms the described baseline.

\section{Conclusion and Future Work}
We proposed a cascaded model using fast RNN-based text classifiers and slow nearest-neighbor based model relying on sophisticated NLP features. We successfully resolved challenges posed by large number of classes, very few training samples per class and slowness of nearest-neighbor approach. We derived a generic lower bound on the accuracy of a 2-stage cascaded model in terms of accuracies of individual stages. We proved a result that eases the effort involved in finding the appropriate number of candidates to pass on to the second stage. We outperformed an LSH-based baseline for query time reduction.

Some problems that need further work naturally emerge. One is investigating insights when 2-stage cascading is extended to multi-stage. Another is exploring other machine learning models operate in a cascaded setting.

\bibliography{references}
\bibliographystyle{splncs03}
\end{document}